\documentclass[sigconf]{acmart}

\renewcommand\footnotetextcopyrightpermission[1]{} 
\usepackage{booktabs} 
\usepackage{tikz}
\usepackage{enumitem}
\usepackage{centernot}
\usepackage[ruled, linesnumbered]{algorithm2e} 

\usepackage{graphicx}

\newtheorem{theorem}{Theorem}[section]

\usepackage{amsthm}

\usepackage{mathtools}
\DeclarePairedDelimiter{\ceil}{\lceil}{\rceil}
\DeclarePairedDelimiter{\floor}{\lfloor}{\rfloor}

\definecolor{Blue}{RGB}{45,47,146}
\setcopyright{none}
\acmConference[KDD'18]{ACM KDD}{August 19-23 2018}{London,
United Kingdom}
\acmYear{2018}
\copyrightyear{2018}

\begin{document}
\title{Community Detection by Information Flow Simulation}

\author{Rajagopal Venkatesaramani and Yevgeniy Vorobeychik}
\affiliation{%
  \institution{Electrical Engineering \& Computer Science\\Vanderbilt University}
  \streetaddress{2201 West End Avenue}
  \city{Nashville} 
  \state{Tennessee} 
  \postcode{37235}
}
\email{{raj.venkat,yevgeniy.vorobeychik}@vanderbilt.edu}


\begin{abstract}
	Community detection remains an important problem in data mining, owing to the lack of scalable algorithms that exploit all aspects of available data - namely the directionality of flow of information and the dynamics thereof. Most existing methods use measures of connectedness in the graphical structure. In this paper, we present a fast, scalable algorithm to detect communities in directed, weighted graph representations of social networks by simulating flow of information through them. By design, our algorithm naturally handles undirected or unweighted networks as well. Our algorithm runs in $\mathcal{O}(|E|)$ time, which is better than most existing work and uses $\mathcal{O}(|E|)$ space and hence scales easily to very large datasets. Finally, we show that our algorithm outperforms the state-of-the-art Markov Clustering Algorithm (MCL) \cite{van2001graph} in both accuracy and scalability on ground truth data (in a number of cases, we can find communities in graphs too large for MCL).
\end{abstract}


\maketitle

\section{Introduction}
	Interest in the study of complex networks in computer science has been on the rise along with the prevalence of internet-based social networking platforms. Much information can be gained from the study of social networks, given that a wide range of information is captured in their structure. At its onset, social networks research dealt primarily with friendships - beginning from Stanley Milgram's paper on the small-world problem \cite{milgram1967}, which revealed the small world effect in modern social networks. As social media evolved to include multiple kinds of interaction, analysis of such networks became increasingly relevant to many stakeholders - including retailers, disaster management and mitigation forces, governments and advertizing agencies to name a few.
	
	In a network, local measures often provide information that is not evident from the network's global view. One local measure of a network's structure is the existence of communities. Although not always well defined, communities are generally accepted to be a subset of the network's nodes that have higher interactions between them as compared to the interactions between two randomly selected nodes from the network as a whole. Detection of such communities in a network has many applications - including but not limited to targeted and viral marketing, information outreach, and influence maximization. Detecting communities is hard and the reasons are manifold. The number of users in a typical social network today is very high, resulting in large datasets. Facebook, for example, has close to 1.9 billion users. Most existing algorithms to detect communities are not scalable to very large networks, owing to their time and space complexities.

	In this paper, we present an algorithm for community detection that is both fast and highly scalable - in terms of time and memory complexity. The approach also models a network using a fully dynamic data structure, instead of more conventional sparse-matrix representations, which contributes significantly to the improvements in performance, and at the same time allows extension of this work to dynamic networks. We find communities by identifying influential users in a network that act as points of origin for information. We then simulate the flow of this information over the network by treating edges as probabilistic paths of information transfer. We provide an upper-bound for our algorithm's expected running time, and illustrate the properties of social networks that we exploit to achieve this bound.
	
\section{Preliminaries}
	A network is represented as a directed, weighted graph, $G = (V, E)$, where $V$ is the vertex-set of $G$ and $E$ is the edge-set of $G$. A directed edge is an ordered pair of vertices, $(u,v)$, where $u, v \in V$, and in our context signifies that information can flow from vertex $u$ to vertex $v$. Each edge has a weight, $w_{uv}$ which is proportional to how probable flow of information is along that edge. We provide a formal interpretation of the same in the next section. 

	\subsection{Definitions and Notation}
	We rely on the following definitions throughout the paper:

		\subsubsection{Node Neighborhood}
		The neighborhood, $N(v)$ of a node $v \in V$ is the set of nodes to which there is an outgoing edge from $v$. \[N(v) = \{x \in V\ |\ (v,x) \in E\}\]

		\subsubsection{Unweighted Out-Degree}
		The unweighted out-degree, $\delta(v)$ of a vertex $v \in V$ is defined as the number of outgoing edges from the vertex. Alternatively, the unweighted out-degree of a node is the cardinality of its neighborhood. \[\delta(v) = |N(v)|\]

		\subsubsection{Unweighted In-Degree}
		The unweighted in-degree, $\delta_i(v)$ of a vertex $v \in V$ is defined as the number of edges in the graph terminating at $v$. \[\delta_i(v) = |\{e=(u,v)\ |\  u \in V, e \in E\}|.\]

		\subsubsection{Weighted Out-Degree}
		The weighted out-degree, $\delta^*(v)$ of a vertex $v \in V$ is defined as the sum of the weights of all outgoing edges from the vertex. \[\delta^*(v) = \sum_{x \in N(v)} w_{vx}\]

		\subsubsection{Graph Diameter}
		The diameter of a graph $G$ is the longest shortest path, i.e. the greatest geodesic distance between any two vertices in $G$.

	\subsection{Properties of Social Networks}
	The following properties of social networks aid in developing our algorithm for community detection. 
	
	\subsubsection{Communities}
	Social networks typically exhibit presence of communities, which are highly connected subsets of the vertex set of the graph. The number of edges within such communities is seen to be higher than those between different communities. 

	\subsubsection{The Small-World Effect}
	The small-world effect is the ability to connect any two vertices of a graph with a path of very short length. A graph, $G$ is said to be \textbf{strongly connected} if \[\forall u, v \in V, \exists \textrm{ a path from } u \textrm{ to } v.\] Most real-life networks contain a strongly connected component that spans most of the graph, often called a \textbf{giant component}. The presence of such a component leads to most vertices being linked by paths of very short length.
	
	\subsubsection{Random Walk}
	A random walk is a stochastic process in which a random walker - from an initial position - moves along a path on a mathematical space in a success of random steps, i.e. chooses one of many possible steps with some probability. In a graph, starting from a vertex, the walker moves to another vertex from the origin's neighborhood with some probability at each step.

\section{Model}
	Rather than using conventional sparse-matrix representations of graph adjacency, we rely on a dynamic data structure to represent the network. The input is a list of triples of the form $(u,v,w_{uv})$ where $u$ and $v$ represent the start and end points of an edge, respectively, and $w_{uv}$ represents its weight. We use a hash table whose keys are the starting vertices $u$. Alternatively, the keys are vertices such that $\delta(u) > 0$. The corresponding value for each key is a 2-tuple. Each element in the tuple, in turn, is a $k$-tuple implemented as a dynamic list, where $k=|N(u)|$. The first of the two $K$-tuples is a list of the vertices that each $u$ is connected to, or alternatively, the node neighborhood, $N(u)$. The second list in the tuple is that of the respective weights of the edges connecting $u$ to the $k$ nodes in its neighborhood. In our implementation, the edge weights are real numbers greater than or equal to $1$. Datasets with edge weights in the range $(0,1)$ should be converted to the appropriate form. The two lists thus, will always be of equal length, and are necessarily ordered. Let $H$ be the hash table over vertices in $G$ with non-zero out-degree. Then for $\forall\ u \in V, \delta(u) > 0$, \[H(u) = (N_u, W_u)\] \[\textrm{where } N_u = (v_1, v_2, \dots v_k)\ ;\;  v_i \in N(u)\] \[\textrm{and } W_u = (w_{uv_1}, w_{uv_2}, \dots, w_{uv_k}).\] Consider the graph in Figure~\ref{graph1}. 

	\begin{figure}[H]
		\vskip -1em
		\begin{tikzpicture}[scale=0.5]
			\filldraw[black] (0,2) circle (2pt) node[anchor=south]{1};
			\filldraw[black] (1.9, 0.6) circle (2pt) node[anchor=west]{2};
			\filldraw[black] (-1.9, 0.6) circle (2pt) node[anchor=east]{3};
			\filldraw[black] (1.17, -1.6) circle (2pt) node[anchor=north]{4};
			\filldraw[black] (-1.17, -1.6) circle (2pt) node[anchor=north]{5};
			\draw[->, black, thick] (0,2) -- (1.85, 0.6) node[midway, above] {\color{Blue}2.0};
			\draw[->, black, thick] (0,2) -- (-1.85, 0.6) node[midway, above] {\color{Blue}4.0};
			\draw[->, black, thick] (0,2) -- (-1.14, -1.5) node[midway, left] {\color{Blue}5.0};
			\draw[->, black, thick] (0,2) -- (1.14, -1.5) node[midway, right] {\color{Blue}4.0};
			\draw[->, black, thick] (-1.9, 0.6) -- (-1.25, -1.55) node[midway, left] {\color{Blue}3.0};
			\draw[->, black, thick] (1.17, -1.6) -- (1.9, 0.5) node[midway, right] {\color{Blue}2.0};
			\draw[->, black, thick] (1.17, -1.6) -- (-1.1, -1.6) node[midway, below] {\color{Blue}1.0};
		\end{tikzpicture}
		\setlength{\belowcaptionskip}{-1em}
		\setlength{\abovecaptionskip}{-2pt}
		\caption{An example graph}
		\label{graph1}
	\end{figure}
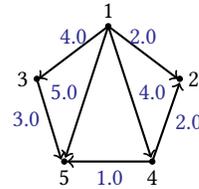

	The hash table representation of the graph is in Figure~\ref{hash1}. For representational clarity, we use square brackets to signify the 2-tuple corresponding to a key. Implementationally, $H[u][0]$ and $H[u][1]$ correspond to $N_u$ and $W_u$ respectively for a given $u$.
	\begin{figure}[H]
	\vskip -1em
		\begin{tabular}{ l | l }
		1 & [ (2, 4, 5, 3) ,  (2.0, 4.0, 5.0, 4.0) ] \\
		3 & [ (5) ,  (3.0) ] \\
		4 & [ (5) ,  (1.0) ] 
		\end{tabular}
		\setlength{\belowcaptionskip}{-1em}
		\caption{Hash Table representation of graph in Figure\ref{graph1}}
		\label{hash1}
	\end{figure}

	\subsection{Implementation Properties}
	This dynamic data structure makes this a fast, scalable implementation. The following properties make this a particularly powerful representation of a network.

	\subsubsection{Complexity}
	Retrieval of neighborhood and corresponding edge weight lists for a node, $u$, is $\mathcal{O}(1)$. Appending a node to any list is an $\mathcal{O}(1)$ operation. Accessing the weight for a certain edge, given start and end vertices $u$ and $v$, is $\mathcal{O}(k)$, where $k=|N(u)|$, and so is searching in $N(u)$. This is comparable to access times in matrix-implementations of graphs. However, the accesses in our algorithm for vertices and respective edge weights are most often in the sequence of their occurrence in the lists, reducing the number of complete list traversals required.

	For each hash-table entry, we need a starting node $u$, a set of end nodes $\{v_1, v_2, \dots, v_n\}$, and corresponding weights for all vertices in $N(u)$. Storing all end-vertices and their corresponding weights requires $2|E|$ space. We only use the vertices in $V(G)$ which have an unweighted out-degree greater than or equal to $1$ as keys. Therefore, \[|\{v \in V, \delta(v)>1\}| \leq {|E|}.\] In addition, during execution, we will store one more list which would contain at most all keys of $H$. Thus worst case total space complexity is \[\mathcal{O}(2|E| + 2|\{v \in V, \delta(v)>1\}|)\] \[ = \mathcal{O}(|E|)\] 

	\subsubsection{Dynamic Nature}
	In conventional matrix-representations of graph adjacency, insertion of a new node is typically only possible before compilation. If possible at runtime, this operation takes $\mathcal{O}(n)$ time where $n$ was the prior number of vertices. In contrast, our implementation can do this in $\mathcal{O}(\delta(v)+\delta_i(v))$ time and space for a new vertex $v$. Given that in most real-life networks, $\delta(v)$ and $\delta_i(v)$ are far less than $|E|$, this speedup is significant. Once a graph has been loaded into memory, our model allows runtime modification. This representation hence, allows extension to dynamic networks.

	\subsection{Reinterpreting Probability on Weighted Edges}
	Over the years, researchers have used random walks for graph clustering by constructing a probability distribution over edge weights at each vertex $u$ as \[P(u,v) = \frac{w_{uv}}{\delta^*(u)}\] where for unweighted graphs, all weights are $1$. This is convenient, because $\sum_{v \in N(u)} P(u,v) = 1\ \forall u \in V$. However, such an interpretation has a few major shortcomings, which we now address.

	\subsubsection{One Path Restriction}
	The classical interpretation is that a random walk chooses to continue on any one of the edges emerging from a vertex with probability \(P(u,v)\) as defined above. However, in reality, information does not flow only along one edge at a time - rather it spreads out on all fronts in parallel. An example is a user making a post, which reaches a number of their followers, in parallel and independently of each other.

	This independence of information flow along multiple fronts is in fact, the driving force behind our approach to community detection. We are no longer restricted to one random walk - rather, we simulate information reaching each of the node's neighbors with some probability in parallel. This probability does depend on edge weights, but as the spread from each edge is independent, the probabilities along all edges emerging from a vertex no longer need to add to $1$. We hence redefine the probability of information flow from vertex $u$ to vertex $v$, along an edge with weight $w_{uv}$ as 
	
	\begin{equation}
	P(u,v) = \left(\frac{w_{uv}}{\delta^*(u)}\right)^\beta; \;\ \  \beta \in (0,1).
	\end{equation}

	This skews individual probabilities closer to 1, in comparison to the classical view, which leads to a better representation of network dynamics. For our experiments, we use a value of $\beta = 1/4$, which captures flow dynamics well in our datasets, as per our interpretation. Thus, for our model, the probability that information flows from vertex $u$ to vertex $v$ is

	\begin{equation}\label{p_uv}
	P(u,v) = \frac{\sqrt[4]{w_{uv}}}{\sqrt[4]{\delta^*(u)}}
	\end{equation}

	A comparison with the classical interpretation follows, illustrating why this is better realization of the data at hand - both for our algorithm and otherwise.

	\subsubsection{Accounting for Disbalance}
	The classical distribution does not account for disbalance in the graph structure - that some parts of the network are better connected than others, and that the dynamics of information flow in these parts are quite different. Consider the following examples.

	\begin{figure}
		\includegraphics[scale=0.65]{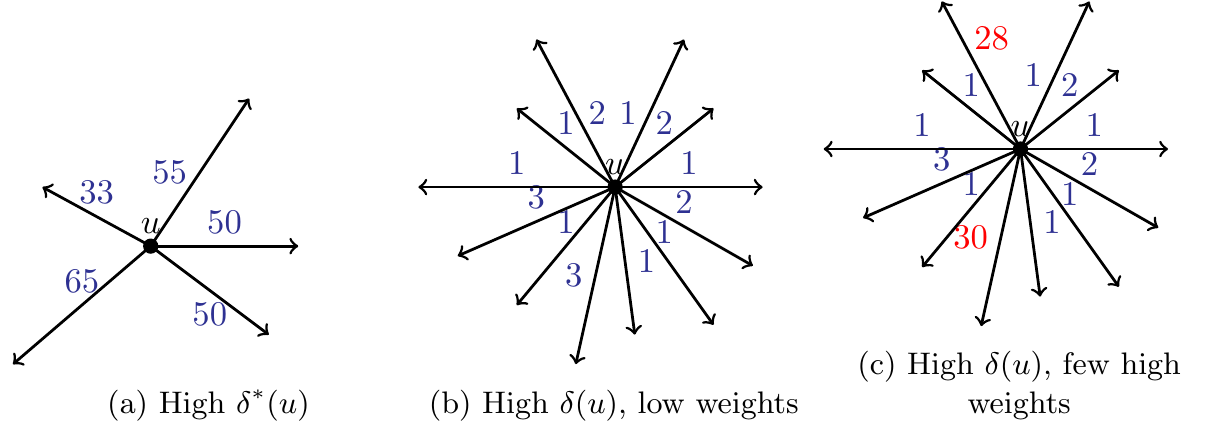}
		\caption{Nodes with varying local structure}
		\label{sub}
	\end{figure}

	\subsubsection*{\textbf{High individual edge weights}}
		Imagine a node where most outgoing edges have high weight with respect to the global average. Figure~\ref{sub}(a) shows one such example. Such a node represents a user from whom information passes to all their neighbors very frequently. By the conventional probability distribution over these edges, information flow along all the edges is roughly equally probable, but a random walk chooses one of the edges at a time. Contrast this with such a user in a real social network - all the vertices in the node neighborhood will likely receive this information in parallel, though not necessarily at the exact same time. Table ~\ref{prob_table1} shows the probabilities for each edge calculated the conventional way, versus by our definition in Equation ~\ref{p_uv}. 

	\begin{table}[htb]
		\vskip -0.5em
		\caption{Probabilities for node with high $\delta^*(u)$}
		\vskip -0.5em
		\begin{tabular}{|c|c|c|}
			\hline
			Edge weight & $w_{uv}/\delta^*(u)$ & $(w_{uv}/\delta^*(u))^{(1/4)}$ \\
			\hline
			45	&	0.169	&	0.641	\\
			50	&	0.188	&	0.659	\\
			55	&	0.207	&	0.675	\\
			65	&	0.245	&	0.703	\\
			\hline
		\end{tabular}
		\label{prob_table1}
	\end{table}

	It is immediately clear why Equation ~\ref{p_uv} is a better interpretation of weights. Our probabilities still reflect the difference in weights, but as each edge now acts as an independent pathway, these edges naturally allow data to flow through them with high probability. This is a major limitation of the typical interpretation - for a node with a large number of highly weighted edges, the individual probabilities over the edges are very low, and do not capture the dynamics of data flow.

	\subsubsection*{\textbf{High $\delta(u)$, low individual edge weights}}
		Now consider a node with a large number of outgoing edges, but with low weights each. Figure ~\ref{sub}(b) shows one such scenario. For such a node, as the number of edges increases, the probability values steeply drop. Using our formulation, probabilities are adjusted according to the local structure. For the node shown in figure ~\ref{sub}(b) the probability values for each edge weight are tabulated in Table ~\ref{prob_table2}.

	\begin{table}[htb]
		\setlength{\abovecaptionskip}{-0.5em}
		\caption{Probabilities for node with high $\delta(u)$, but low $\delta^*(u)$}
		\vskip 1em
		\begin{tabular}{|c|c|c|}
			\hline
			Edge weight & $w_{uv}/\delta^*(u)$ & $(w_{uv}/\delta^*(u))^{(1/4)}$ \\
			\hline
			1	&	0.052	&	0.478	\\
			2	&	0.105	&	0.569	\\
			3	&	0.157	&	0.630	\\
			\hline
		\end{tabular}
		\label{prob_table2}
	\end{table}

	Given that nodes with higher $\delta^*(u)$ will be much more active, one may argue that the probabilities calculated for this node are too high. However, they are implicitly conditional on the information creation rate. Once a post has been created (with a possibly lower rate), the data would reach its neighborhood with a higher probability than what the conventional interpretation suggests. We now consider another variation on this type of node - when there are a large number of edges with low weights, and few edges with high edge weights. Consider the node in Figure ~\ref{sub}(c) , where high edge weights are highlighted in red. The probability values for this node are presented in Table ~\ref{prob_table3}.

	\begin{table}[H]
		\caption{Probabilities for node with high $\delta(u)$, and some high weight edges}
		\begin{tabular}{|c|c|c|}
			\hline
			Edge weight & $w_{uv}/\delta^*(u)$ & $(w_{uv}/\delta^*(u))^{(1/4)}$ \\
			\hline
			1	&	0.013	&	0.343	\\
			2	&	0.026	&	0.408	\\
			3	&	0.039	&	0.451	\\
			28	&	0.388	&	0.789	\\
			30	&	0.416	&	0.803	\\
			\hline
		\end{tabular}
		\label{prob_table3}
	\end{table}

	It is evident that the edges with weights $28$ and $30$ will be high-flow paths and yet the conventional interpretation assigns them probabilities of less than half. Given that information does flow in parallel in real networks, both those nodes are almost equally likely to have data flow through them with high probability, which our method clearly reflects.

\section{Method}
Detecting communities in a graph, $G$ is defined as finding a set of $K$ subgraphs, $S=\{S_1, S_2, \dots, S_K\}$ of $G$, such that they are highly connected. Solving this problem optimally for any given metric of connectedness is $\mathcal{NP}$-hard, and hence solutions are approximations. The basic idea behind our approach is to find influential users in a network, and then simulate flow of information originating at those nodes, with probabilities calculated over edge weights. As this information spreads throughout the graph, we propagate cluster labels till either the whole graph is covered, or the algorithm reaches a stable state where there is no flow of data for some $\lambda$ iterations. Trivial assumptions allow us to use undirected and unweighted graphs.

Unlike all existing work on community detection, our approach does not guarantee a cluster label for each node. This is not a limitation, but a conscious design choice. In any real network, there are dormant users who do not create/propagate information regularly. In theory, for a connected graph, it is possible to use our approach to assign cluster labels to all vertices - by removing the stable-state termination condition, and let the algorithm converge probabilistically, but that then may not be an accurate representation of community structure. While evaluating cluster assignments, we treat these dormant nodes as singleton clusters. We also do not formulate our approach as an optimization problem over a single metric. Many metrics defined for connectedness have certain shortcomings, and by solely maximizing/minimizing those metrics, we tend to end up with cluster assignments that reflect these limitations. Instead, we simply let labels propagate through the network, and then evaluate the assignments thus found using different metrics to see how we fare.

\subsection{Alpha Detection}
We begin by identifying users that are highly likely to be sources of information in a social network. Such users typically have a large number of friends/followers, i.e. a large $\delta(u)$ value, as well as these connections tend to be strong, which translate to high edge weights, i.e. a large $\delta^*(u)$ value. Such users are referred to as \textit{alpha} users. Existence of such users has been established in existing literature \cite{barabasi1999emergence, muchnik2013origins, albert2002statistical, lattanzi2015power}, and is proven in Lemma ~\ref{min_max_degree}. To find these alpha users, we follow the steps in Algorithm ~\ref{alpha_detect}.

\begin{algorithm}[ht]
\SetKwInOut{Input}{Input}\SetKwInOut{Output}{Output}
\DontPrintSemicolon
\SetKwFunction{Sort}{Sort}
\SetKwFunction{Truncate}{Truncate}
	\Input{$G=(V,E)$\\$k$, Percentage of vertices to consider for alphas}
	\Output{Set $X$ of alpha vertices}
	\BlankLine
	List $NumRank \longleftarrow V$\;
	List $DegRank \longleftarrow V$\;
	\Sort{$NumRank$} by $\delta(u)$\;
	\Sort{$DegRank$} by $\delta^*(u)$\;
	\Truncate{$NumRank$} to retain top $k \%$ elements\;
	\Truncate{$DegRank$} to retain top $k \%$ elements\;
	$X \longleftarrow NumRank\ \cap\ DegRank$
\caption{Alpha Detection}
\label{alpha_detect}
\end{algorithm}

We have now found a subset of users in the network who have both high $\delta(u)$, as well as high $\delta^*(u)$. Note that the value of $k$ in the algorithm needs to be decided by the user. Finding the optimal number of clusters in a network is a separate, hard problem in itself, that we do not attempt to solve in this paper. Rather, we make informed guesses for values of $k$, depending on the dataset at hand. 

\section{Algorithm}
\subsection{Approach} The intuition behind our approach is that communities are built around influential users in a social network. Having found such alphas in the previous steps, we now imagine these vertices to be points of information origin. We then simulate how this data propagates through the network, with the idea that for a vertex $u$, if $\alpha_i$ is the first alpha vertex $u$ receives data from, $u$ belongs to the $i^{th}$ community. Therefore, we end up with $|X|$ communities. If for a vertex $x \in V$, there is no path from any alpha node to $x$, then $x$ will not be assigned a community.

The outline of the algorithm is as follows. We assign a cluster label to each alpha. Then in parallel for each alpha, we spread the label to each vertex its neighborhood, with a probability as calculated in Equation ~\ref{p_uv}. In the next step, for each node visited in the previous iteration, spread the acquired label to its neighborhood. We consider a vertex visited if it has been assigned a label. For each iteration, the edges along which data did not previously flow continue to remain active paths - i.e. information may flow through them in a later iteration. Nodes with neighborhoods that are not completely visited, hence, remain active. This process continues until either labels have been assigned to all nodes in the graph, or until a stable state is reached. A value of $\lambda = 3$ for steady-state termination was used for all experiments presented in this paper. 

At each iteration, we divide the active paths into subsets for parallelization. The number of these subsets depends on hardware used, and should be chosen keeping in mind that too many subsets may lead to increased thread waiting times. Algorithm ~\ref{CommDet++} explains the implementation of our approach. 

\begin{algorithm}[htb]
\SetKwInOut{Input}{Input}\SetKwInOut{Output}{Output}\SetKwInOut{Initialize}{Initialize}
\DontPrintSemicolon
\SetKwFunction{Append}{Append}
\SetKwFunction{Delete}{Delete}
\SetKwFor{ForAll}{for all}{do}{end}
	\Input{$G=(V,E)$\\$X$, set of alphas\\$\lambda$, max iterations with no data flow}
	\Output{hash table $Comm$ of cluster assignments}
	\Initialize{List $Origin \longleftarrow X$ \\List $Visited \longleftarrow \emptyset$ \\Hash Table $Comm \longleftarrow (v,0)\ \forall\ v\ \in\ V$ \\$add = 0\hskip 3cm$  //Counter Variable\\}
	\BlankLine
	\ForAll {$u\ \in\ Origin$} {
		$Comm[u] \longleftarrow u$\;
	}
	\While{$add \leq \lambda$} {
		$add \longleftarrow add + 1$\;
		\ForAll {$u\ \in\ Origin$ \textbf{in parallel}} {
			\If {$N(u) \neq \emptyset$} {
				\ForAll {$v\ \in\ N(u)$} {
					\If {$Comm[v]==0$, i.e. $v$ is not visited}{
						Propagate label $u$ to $v$ with probability \[P(u,v) = \frac{\sqrt[4]{w_{uv}}}{\sqrt[4]{\delta^*(u)}}\]\;
						\If {Label Propagated} {
							\Append{$Origin$, $v$}\;
							$Comm[v] \longleftarrow u$\;
							$add=0$\;
						}
					}
				}
				\If {$Comm[x] \neq 0\ \forall\ x\ \in\ N(u)$}{
					\Delete{$Origin$, $u$}\;
					\Append{$Visited$, $u$}\;
				}
			}
			\Else {
				\Delete{$Origin$, $u$}\;
				\Append{$Visited$, $u$}\;
			}
		}
	}
\caption{Label Propagation}
\label{CommDet++}
\end{algorithm}

\subsection{Running Time Upper Bound}
\begin{lemma}
\label{E_runtime}
For $n$ independent Bernoulli events with probabilities $p_1, p_2, \dots, p_n$ tried in parallel, the expected number of trials required to reach success on all $n$ events at least once is $\mathbb{E}=\frac{1}{\min p_i, i \in [1,n], i \in \mathbb{Z}}$.
\end{lemma}
\begin{proof}
Consider a Bernoulli event with probability $p$. Let $X$ be a random variable denoting the number of independent trials to success. Then \[\mathbb{E}(X) = p + 2(1-p)p + 3(1-p)^2p + \dots\]
\[(1-p)\mathbb{E}(X) = p(1-p) + 2p(1-p)^2 + \dots\]
\[\mathbb{E}(X) - (1-p)\mathbb{E}(X) = p + p(1-p) + p(1-p)^2 + \dots\]
\[p\mathbb{E}(X) = p + p(1-p) + p(1-p)^2 + \dots\]
\[\mathbb{E}(X) = 1 + (1-p) + p(1-p)^2 + \dots\] This is a Geometric Series with the common ratio $(1-p)$, and thus, \[\mathbb{E}(X)=\frac{1}{1-(1-p)} = \frac{1}{p}\]
Now consider $n$ such Bernoulli events with probabilities $p_1, p_2, \dots, p_n$ tried in parallel. Then, for the event $i$ with the lowest probability, $p_i, \mathbb{E}(X_i)\geq\mathbb{E}(X_k)\ \forall\ k \in [1,n]$, and the expected time to first success on all events is at most $\mathbb{E}(X_i) = \frac{1}{\min p_i, i \in [1,n], i \in \mathbb{Z}}$.
\end{proof}

\begin{lemma}
\label{diameter}
In a connected graph, $G=(V,E)$ with diameter $d$, starting from any vertex $u \in V$, the entire graph can be visited in $d$ successive neighborhood traversals.
\end{lemma}
\begin{proof}
As $G$ is connected, $\exists\ \textrm{ a path from } u \textrm{ to } v\ \forall\ u,v\ \in V$. The $i^{th}$ neighborhood of a vertex, $u$ is \[\{v \in V\ |\ \textrm{ length of shortest path from } u \textrm{ to } v \textrm{ is } i\}\] Then $\max i = d$, and the entire graph is visited.
\end{proof}

\begin{lemma}
\label{min_max_degree}
Let $G=(V, E)$ be a graph with $n$ vertices and a diameter $d$, and minimum possible value of maximum vertex degree.
\begin{enumerate}
	\item The optimal maximum degree is $\ceil*{n^{\frac{1}{\floor*{d/2}}}}+1$.
	\item $G$ has at least $n - \left((k-1)^{\floor*{\frac{d}{2}}} - 1\right) \times k$ vertices of degree at least $k$, where $k = \ceil*{n^{\frac{1}{\floor*{d/2}}}}+1$.
\end{enumerate}
\end{lemma}
\begin{proof}
	\textbf{1.} We construct a graph such that it minimizes maximum vertex degree subject to the constraints $|V|=n$ and $diam(G) = d$. We can use at most $d+1$ vertices to form a maximal path of length $d$, with the maximum vertex degree $2$. At this point, all vertices except the ones at the end have degree $2$, and we cannot add a vertex to the ends of the path, as we are bound by $diam(G) = d$. If we now allow maximum degree to increase by $1$, we can add a chain of $t$ vertices to each internal vertex, where $t$ is the shortest distance from the internal vertex to the nearest end-vertex of the original $d$-path. The end of each newly added chain is now at a distance $d$ from all end vertices. Then to each internal node in the newly added chains, we can recursively add chains so that the maximum distance is preserved, till no new vertices can be added to the graph.

	For a graph with given diameter $d$ and a maximum degree $k$, there is an upper bound to the number of vertices one can add. Equivalently, we can have at most $k$ trees, each with a branching factor of $k-1$, and the $k$ roots connected to a common vertex. This is justified as follows - $k$ roots are connected to a common vertex, which now has the maximum degree. Each root can then have a maximum branching factor of $k-1$, as the maximum degree of $G$ is $k$, and this constraint recursively applies to all children. The maximum depth of each tree of the $k$ trees can not exceed $\floor*{\frac{d}{2} - 1}$, due to the diameter constraint. The floor function accounts for odd diameters.

	To solve for the minimum required maximum degree, we need the total possible number of nodes in such a graph to be greater than or equal to the number of nodes we want to construct a graph with. For a given $k$ and $d$, we know that we can construct at most $k$ trees of height $\frac{d}{2} - 1$ and branching factor $k-1$. Thus, maximum number of nodes possible for given $k$ and $d$ is given by the inequality
	\[\textrm{Max. No. of trees} * \textrm{Max. nodes in each tree} \geq n-1 \]
	\begin{equation}\label{n_trees}
		k * \frac{(k-1)^{\floor*{(d/2)-1}+1} - 1}{k-2} \geq n-1
	\end{equation}
	For small values of $k$, solving this inequality is easy, and provides the optimum $k$. For large values of $k$, $k/(k-2)$ converges to $1$. Thus we use the approximation
	\[ (k-1)^{\floor*{(d/2)}} - 1 \geq n-1  \]
	\[ (k-1)^{\floor*{(d/2)}} \geq n  \]
	\[ k-1 \geq n^{\frac{1}{\floor*{d/2}}} \]
	\[ k = \ceil*{n^{\frac{1}{\floor*{d/2}}}} + 1 \]

	For small $k$, where $k/(k-2)$ is not close to $1$, this approximation leads to an error of $1$. Solving Equation ~\ref{n_trees} directly, however, leads to the correct result. The node to which all roots of the balanced trees are connected has degree $k$, and thus $G$ has at least one node with degree $k$. Removing an edge, and linking the node to any other vertex will either violate $diam(G)=d$ or increase maximum degree to $k+1$.

	\textbf{2.} When maximum degree is minimized, for a given $k$ and $d$, the maximum number of leaves possible are 
	\[ \left((k-1)^{\floor*{\frac{d}{2}}} - 1\right) \times k \]
	If the number of vertices available are less than the maximum possible in the construction, to keep $G$ connected, we can only remove a leaf for each, each of which reduces the number of vertices with degree at least $k$ by $1$. Thus, $G$ has at least \[n - \left((k-1)^{\floor*{\frac{d}{2}}} - 1\right) \times k\] vertices with degree at least $k$.
\end{proof}

\begin{lemma}
\label{O_E}
	For a social network, $|V| log |V| = \mathcal{O}(|E|)$.
\end{lemma}
\begin{proof}
	The clustering coefficient of a graph is given as the ratio of the number of triangles to the number of 2-paths in the graph, where a triangle is a fully connected subgraph on 3 vertices. The Barab\'asi-Albert (BA) model produces random graphs that follow a power-law degree distribution. Social networks have been shown to follow such power-law distributions \cite{barabasi1999emergence, muchnik2013origins, albert2002statistical, lattanzi2015power}, in which the probability of observing a vertex with degree $i$ is proportional to $i^{-\alpha}$ for some $i>i_{min}$. Though real-world networks show much higher clustering than those generated by the BA model, these suffice to provide a lower bound on the number of edges in $G$. The BA model yields a clustering coefficient
	\[CC(G) = \frac{(\ln n)^2}{n}\ \ \textrm{ where } n =|V|.\]
	The network average of local clustering coefficient is a reasonable approximation to the global clustering coefficient ~\cite{watts1998collective}. Then at each vertex $u$, local clustering coefficient $CC(u) \approx CC(G)$. In addition, the local clustering coefficient at $u$ is defined as the ratio of the number of edges in $N(u)$ to the number of possible edges with $|N(u)|$ vertices. Equating the two,
	\[ \frac{2|\{(u,v)\in E\ |\ v \in N(u)\}|}{|N(u)|*(|N(u)-1|)} = \frac{(\ln |N(u)|)^2}{|N(u)|} \]
	\begin{equation}\label{cc}
	|\{(u,v)\in E\ |\ v \in N(u)\}| = \frac{|N(u)|(\ln |N(u))^2}{2}
	\end{equation}
	The graph we previously constructed has a clustering coefficient of $0$. We take this to be the worst case for a social network. For such a graph to have a global clustering coefficient similar to networks that follow a power-law distribution, we must add edges to the graph such that the average number of edges at each node is similar to that in Equation ~\ref{cc}. Therefore, averaging over all vertices, $|E|\geq c |V| (\ln |V|)^2$, and $|V| \log |V| = \mathcal{O}(|E|)$.
\end{proof}

\begin{theorem}
The expected running time of our algorithm in the worst case is $\mathcal{O}(|E|)$.
\end{theorem}
\begin{proof}
Let $=(V,E)$ be a graph representing a social network. Sorting the vertices for alpha detection takes $|V| log |V|$ time, which by Lemma ~\ref{O_E} is $\mathcal{O}(|E|)$. Set intersection on a sorted list is linear in the size of the list, and hence the entire alpha detection process is $\mathcal{O}(|E|)$.

Cluster assignments are only done in the largest connected component of $G$. The worst case is serial execution on each neighborhood, and not reaching a steady state till $G$ is covered. At each step $i$, the $i^{th}$ neighborhood is traversed serially. For a starting vertex, $u$, for each $i$, let the number of edges between the $(i-1)^{th}$ and the $i^{th}$ neighborhoods of $u$ be $X_i$. In the worst case, each edge $(v,w)$ in $X_i$ requires $\mathbb{E} = \frac{1}{\min P_{vw}, (v,w) \in X_i}$ trials by Lemma ~\ref{E_runtime}, which is a constant for each $i$. Thus, iteration $i$ takes $\mathbb{E} * |X_i|$ time. By Lemma ~\ref{diameter}, $\max i = d$, where $d$ is the diameter of $G$. As $G$ is a social network, it exhibits the small-world effect and $d$ is a small positive constant. Summing over all iterations,
\[T = \sum_{i \in [1,d], i \in \mathbb{Z}}\frac{1}{\min P_{vw}, (v,w) \in X_i} |X_i|\]
\[T = \sum_{i \in [1,d], i \in \mathbb{Z}}c\ |E|\]
\[T = d\ c\ |E|\]
\[T = \mathcal{O}(|E|) \]
\end{proof}

\section{Experiments}
\subsection{Data}
We use 14 datasets to compare our approach to MCL. Notably, the Karate Club network \cite{zachary1977information} has been used as a standard test set for most community detection algorithms. The said network is unweighted, undirected - so we convert it to a weighted directed network, by assuming each edge to have a weight of $1$, and assuming each undirected edge stands for two edges, one in each direction. As our algorithm only considers outgoing edges for calculations, loops are automatically ignored. The `correct' split of the network into $2$ communities is known, as the dataset is a representation of an actual split in a University Karate Club. 

Secondly, we use a Facebook interaction network, which is a weighted, directed graph in which each weighted edge represents the number of times the user represented by the origin vertex posts a message to the wall of the one represented by the destination vertex of the edge, in a given time frame. This dataset was made available by \cite{viswanath-2009-activity, konect}.

We use 4 datasets from the Stanford Network Analysis Project (SNAP) \cite{snapnets}, that include ground-truth communities. The DBLP network \cite{yang2015defining} is an undirected unweighted scientific co-authorship network, where authors are connected if they have published at least one paper together. The YouTube social network dataset \cite{yang2015defining} \cite{mislove-2007-socialnetworks} shows user friendships on the service, and is undirected and unweighted. The Wikipedia network of categories \cite{yin2017local} \cite{KlGlKo14} is a directed, unweighted network of hyperlinks on Wikipedia articles that belong to categories with at least 100 articles in them. The categories serve as ground-truth communities. The email EU Core network \cite{leskovec2007graph} \cite{yin2017local} is a network of e-mail exchanges between persons belonging to a large European research institution. This network is directed, unweighted. The DBLP, YouTube and Wikipedia networks are very large networks, with 317080, 1134890 and 1791489 vertices respectively, and 1049866, 2987624 and 28511807 edges respectively, to which our algorithm successfully scales.

The other 8 datasets were made available by Tore Opsahl. These include a Facebook-like social network \cite{opsahl2009clustering}, a Facebook-like forum network \cite{opsahl2013triadic}, Freeman's EIES network \cite{freeman1979networkers}, organizational advice, awareness and consultation networks \cite{cross2004hidden}, a scientific collaboration network (Arxiv) \cite{newman2001structure} and a women's club network \cite{pajek}. All of these are weighted networks of varying sizes, of which 3 are undirected.

\subsection{Evaluation}
Though improvements to community detection based on the Markov Clustering Algorithm have been made, MCL itself remains one of the most widely-used methods with applications ranging from marketing \cite{hong2013database,10.1007/978-3-642-40270-8_10} to biological sciences \cite{udrescu2016clustering, zhang2017therapeutic, alawieh2015rich}. It is also the most similar in nature to our approach - the difference being that MCL simulates information flow using a single random walk to capture community structure, whereas our approach is somewhat like spawning multiple random walks on each alpha node. When ground truth information is available, we compare false positive and negative rates against MCL. For other datasets, we compare plot conductance of clusters against cluster sizes. Conductance $\phi$ of a cluster, $S$ is defined as \[\phi = \frac{c_S}{2m_S + c_S}\] where $c_S = |{(u,v), u \in S, v \not\in S}|$. In other words, $c_S$ is the number of boundary edges of the cluster, and $m_S$ is the number of edges within the cluster. A lower conductance score indicates better clustering quality. We also compare running time with MCL for some of the datasets used.

\section{Results and Discussion}
\subsection{Ground-Truth}
We begin with the Karate Club Network dataset, where the true split over 2 communities is known. MCL consistently misclassifies 2 nodes in this network, and our approach misclassifies 1-2 nodes in most test runs, owing to the randomness-based implementation of edge-traversal probabilities, but also arrives at the correct split in some cases.

For the large datasets with ground-truth, checking all possible pairs of vertices is not feasible due to their size, so we randomly sample equal number of positive and negative cluster assignments and compare it with the ground truth till the overall false positive and false negative rates converge. The results presented for our approach are for values of $k$ which were arrived at experimentally, or by making informed guesses. Additionally, the ground truth cluster assignments are overlapping, so we assume an assignment for vertices $u$ and $v$ is correct if they are assigned the same cluster in our results, and are in the same cluster in any of the possible clusterings in the ground truth. Figure ~\ref{GroundTruth} presents these results.

\begin{figure}[H]
	\centering
	\includegraphics[bb=0 0 640 452, scale=0.4]{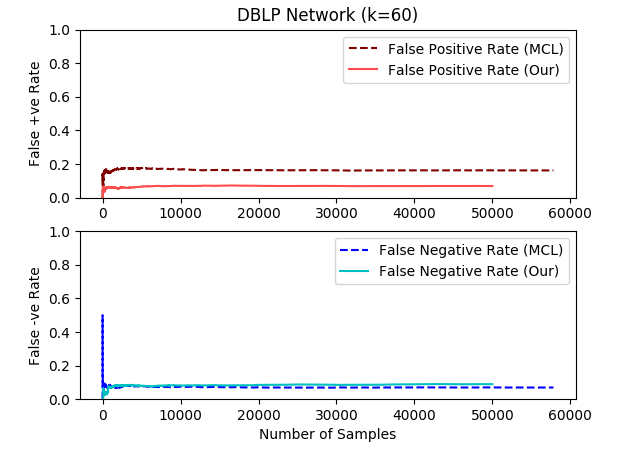}

	\includegraphics[bb=0 0 640 452, scale=0.4]{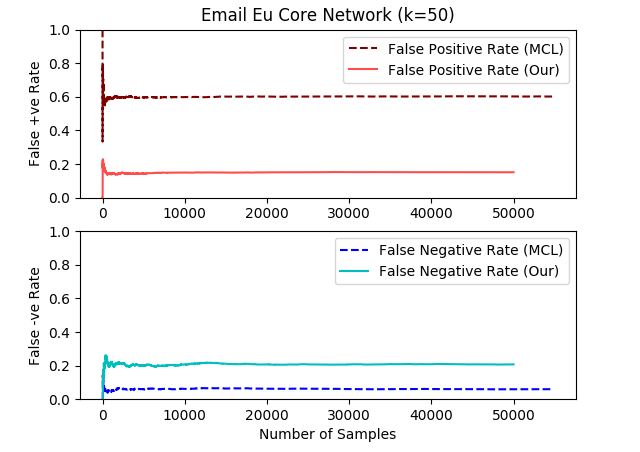}

	\includegraphics[bb=0 0 640 433, scale=0.4]{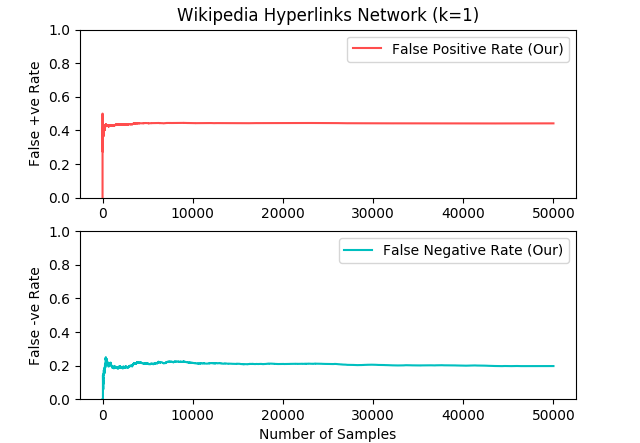}
    \phantomcaption
\end{figure}
\begin{figure}
	\ContinuedFloat
	\includegraphics[bb=0 0 640 452, scale=0.4]{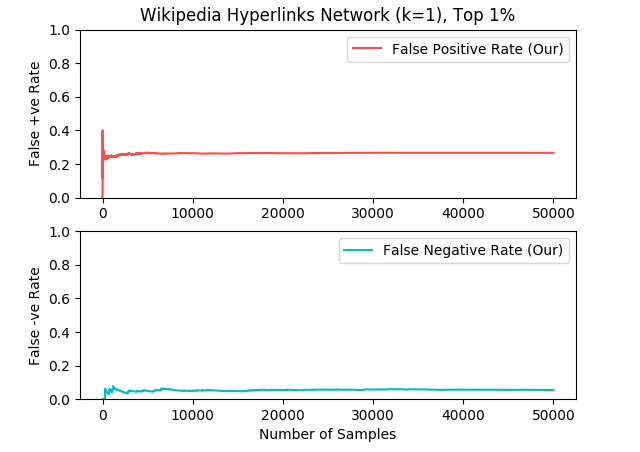}
	\vskip 1em
	\includegraphics[bb=0 0 640 452, scale=0.4]{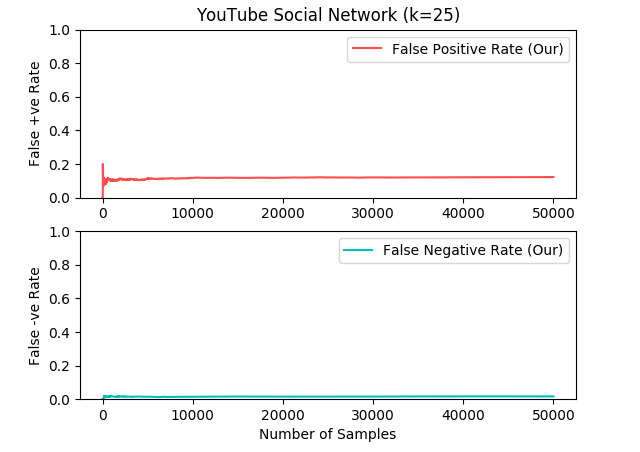}
	\caption{Ground Truth Evaluations}
    \label{GroundTruth}
\end{figure}

We observe that our approach does significantly better in terms of false positive rates as compared to MCL, whereas it performs slightly worse on false negative rates for both the DBLP and Email Eu Core datasets. The difference in false negative rates on the DBLP dataset is negligible. The difference in false positive rates on the Email Eu Core network dataset is stark - the MCL converges to a false positive rate of about $0.6$, whereas our rate converges to about $0.18$. MCL is unable to scale to the Wikipedia and YouTube datasets, and hence we present absolute results. In the Wikipedia network, our false positive rate initially converges to about $0.475$. This lapse in performance, however, is attributed to the quality ground-truth communities provided. The two largest ground truth communities are the categories of Wikipedia articles ``Living People'' and ``Year of Birth Missing (Living People)''. In both these categories, two randomly chosen articles are not likely to have a link to each other in the hyperlink-network. To compensate for this, we remove these two categories from the ground truth data. We then bias the positive sampling towards nodes from the top 1\% largest clusters, and the negative samples towards nodes that either lie in the smallest clusters or that do not lie in clusters. With this bias, our false positive rate falls to about $0.25$ and false negative rate falls to less than $0.1$. These biased rates are not necessarily indicative of global average - they only reflect good performance on detecting the most well-connected clusters.

The results on the YouTube dataset are particularly interesting, and indicative of our algorithm's success. Historically, the YouTube dataset has been extremely tricky for community detection, because its communities do not exhibit behavior as quantified by standard metrics. Ten community detection algorithms published after 2010 have been compared in \cite{harenberg2014community}, none of which perform well on the YouTube dataset. Our approach however, is able to find the underlying community structure, which is very close to the ground truth. The false positive rate converges to about $0.15$ and the false negative rate converges to less than $0.02$ with an unbiased sample.

\subsection{Conductance}
Due to space constraints, we provide the conductance profiles for only certain datasets, but a similar trend is seen over all 14 datasets used - for high value of $k$, the conductance over our assignments are closer to the results from MCL, whereas for lower values of $k$, our algorithm finds a larger number of large clusters, with fairly low conductance. Figure ~\ref{conductance} presents the conductance profiles by cluster sizes.
\begin{figure}[H]
	\vskip -1em
	\includegraphics[bb=0 0 600 600, scale=0.4]{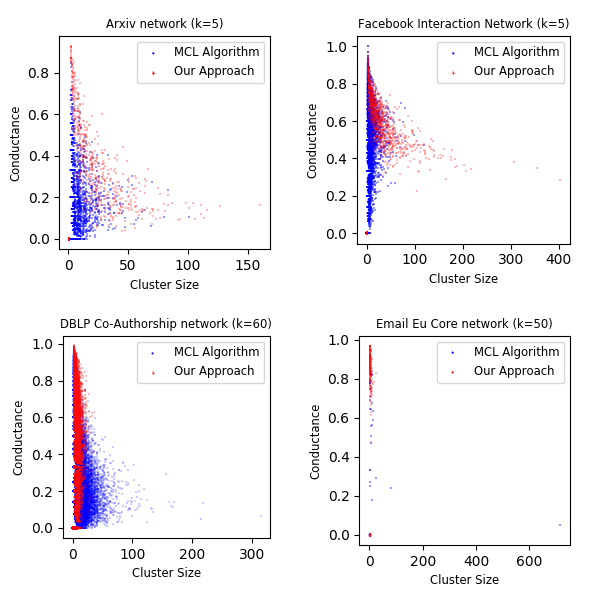}
	\setlength{\abovecaptionskip}{-1em}
	\setlength{\belowcaptionskip}{-1em}
	\caption{Conductance Profiles}
	\label{conductance}
\end{figure}

Typically, lower conductance scores have been associated with better clustering quality in networks. Our ground-truth experiments show that our approach is superior to MCL on the DBLP and Email Eu Core networks. Yet, the conductance profiles indicate that our conductance scores are in fact worse than those obtained by MCL. This shows that conductance in isolation is not a good metric to quantify community structure. 

\subsection{Runtime}
The following table compares the running time for 9 out of the 14 datasets used. The other 5 datasets are very small, so comparing runtimes may not be meaningful. All runtimes were tested on a machine with an Intel i7-5500U and 8GB DDR3 RAM, running Ubuntu 17.04. For MCL, we use the author's original implementation, provided as a compiled executable. Our code is in Python 3, and further speedups are expected if the algorithm is implemented in a lower level language.

\begin{table}[htb]
	\caption{Running times for our approach vs. MCL}
	\begin{tabular}{|c|c|c|c|}
		\hline
		Dataset & No. of Edges & MCL & Our Approach\\
		\hline
		EIES network 		&	830		&	0.009s			&	0.081s	\\
		Fb-like Social 		&	20296	&	0.924s			&	0.234s	\\
		Email				&	25571	&	0.810s			&	0.190s	\\
		Arxiv				&	95188	&	1.957s			&	1.102 s	\\
		Fb-like forum 		&	142760	&	0.949s			&	0.492s	\\
		FB-Interaction		&	253831	&	20.643s			&	4.903s	\\
		DBLP				&	1.05 M	&	8m 33s			&	6m 28s	\\
		YouTube				&	5.98 M	&	Does not scale	&	1hr 7min	\\
		Wikipedia			&	28.5 M	&	Does not scale	&	3hr 4min	\\
		\hline
	\end{tabular}
	\label{prob_table1}
\end{table}

These results indicate that the running time for our algorithm is almost linear in the number of edges, which is in accordance with our derived upper bound as per Theorem ~\ref{O_E}. This makes our algorithm one of the fastest algorithms for community detection. In the near future, we aim to further improve runtimes by optimizing parallelization.

\section{Limitations and Future Work}
A major drawback of our approach is that it requires the parameter $k$ to be decided by the user. Intelligent guesses can however be made, given the nature of the data at hand. In a co-authorship network, for instance, irrespective of the total number of users, it is highly likely that most users in the network have only collaborated with a small subset of the users. Thus, for such detecting a large number of small communities, we would use a larger value of $k$. 

On the other hand, when ground truth information is not available, this parameter may also be used to arrive at different possible clusterings, depending on the task at hand. Depending on the distribution over cluster sizes desired, one can use appropriate values of $k$ to produce such clusterings. 

Future extensions to this work will potentially include real-time models of dynamic networks, models for heterogeneous and multi-layered networks, and their applications.

\begin{acks}
	The authors are grateful to Dr. Sudeepto Bhattacharya (Associate Professor, Department of Mathematics, Shiv Nadar University), discussions with whom planted the seeds of this work.
\end{acks}

\bibliographystyle{ACM-Reference-Format}
\bibliography{biblio} 

\end{document}